\documentclass[a4paper]{llncs}

\usepackage{algorithm}
\usepackage[noend]{algorithmic}
\usepackage{amsmath}
\usepackage{amssymb}
\usepackage{graphicx}
\usepackage{dsfont}
\usepackage{setspace}

\spnewtheorem{observation}{Observation}{\bfseries}{\itshape}
\spnewtheorem{myclaim}{Claim}{\bfseries}{\itshape}{\rmfamily}
\spnewtheorem{myproblem}{Problem}{\bfseries}{\itshape}

\newcommand{\old}[1]{{}}

\newcommand{\smallt}{$2 \leq t <2 \frac{1}{5}$}
\newcommand{\bigt}{$t \geq 2 \frac{1}{5}$}

\let\doendproof\endproof
\renewcommand\endproof{~\hfill\qed\doendproof}

\title{ Minimum Weight Euclidean $t$-spanner is NP-Hard
\thanks{Research is partially supported by the Lynn and William Frankel
Center for Computer Science and by
grant 680/11 from the Israel Science Foundation (ISF). }}

\author{Paz Carmi \and Lilach Chaitman-Yerushalmi}
\institute{ Department of Computer Science,\\ Ben-Gurion University of the Negev, Israel}

\begin{document}
\maketitle


\begin{abstract}

Given a set $P$ of points in the plane, an Euclidean $t$-spanner for $P$ is a 
geometric graph that preserves the Euclidean distances between 
every pair of points in $P$ up to a constant factor $t$.
The weight of a geometric graph refers to the total length of its edges.
In this paper we show that the problem of deciding whether there exists 
an Euclidean $t$-spanner, for a given set of points in the plane, of weight at most $w$ is NP-hard
for every real constant $t>1$, both whether planarity of the $t$-spanner is required or not.

\end{abstract}

\section{Introduction}\label{sec:Intro}


Consider a weighted graph $G=(V,E)$ with a weight function $w:E \rightarrow \mathds{R}^+$ over the edges.
For any two vertices $u,v \in V$, we denote the weight of the shortest path between $u$ and $v$ in $G$ by $\delta_G(u,v)$.
Given a spanning subgraph $G'$ of $G$, we define the \emph{dilation} of $G'$ (with respect to $G$) to be the value
$$ \max_{u,v \in V} \frac{\delta_{G'}(u,v)}{\delta_{G}(u,v)}.$$
Given a real value $t> 1$, a $t$-spanner of $G$ is a spanning subgraph $G'$ 
with dilation at most $t$ with respect to $G$.
Thus, the the shortest-path distances in $G'$ approximate shortest-path distances
in the underlying graph $G$ within an approximation ratio $t$.
Typically, $G$ is a dense graph with $\Omega(n^2)$ edges and 
the $t$-spanner $G'$ is desired to be sparse, preferably having only a linear number
of edges.


Spanners have been studied in many different settings. 
The various settings differ from one another in the characterization of the underlying graph $G$,
such as different topologies and specific weight functions over the edges,  
in the value of the required dilation $t$, and in the properties of the spanner $G'$, such as planarity.
We concentrate on the setting where the underlying graph is geometric.
In our context, a graph $G=(P,E)$ is called geometric graph or Euclidean graph
if its vertex set $P$ is a set of points in the plane 
and every edge $\{p,q\}$ in $E$ is the line segment $\overline{pq}$, 
weighted by the Euclidean distance between its endpoints $|pq|$.
Moreover, the underlying graphs we consider are complete graphs
and therefore we refer to their $t$-spanners as spanners of the point set $P$.
There is a vast body of literature on $t$-spanners in this geometric setting 
(see~\cite{GiriSmid07} for a comprehensive survey of the area).


The weight of a geometric graph is defined as the sum of the lengths of its edges. 
The weight is a good measure of the cost of building the network;
thus, it is often desirable to have spanners with low weight.
Since any spanner must connect all the points, the weight of a $t$-spanner is bounded 
from below by the weight of a minimum spanning tree $MST(P)$.
Chandra et al.~\cite{Chandra} have presented a greedy algorithm for constructing 
a $t$-spanner in $O(n^3 \log n)$ time, which has been proved 
by Das et al.~\cite{DasHN93,Das2} to have a weight of size $O(wt(MST(P)))$. 
The constant factor depends on the value $t$. 
A more efficient algorithm that computes the greedy spanner in $O(n^2\log n)$ 
was later developed by Bose et al.~\cite{BCFMS08}.
A fast implementation of a variant of the greedy algorithm that maintains the $O(wt(MST(P)))$ weight
and runs in $O(n \log^2 n)$ time has been developed by Das and Narasimhan~\cite{DN97}.
Those spanners approximate the weight of the minimum $t$-spanners within a constant factor that
depends on $t$, but they are not necessarily optimal.
We address the following decision problem and appropriate optimization problem.

\begin{myproblem}\label{prob:LWS$t$}
The \emph{\textsc{Low Weight $t$-Spanner} (LWS$t$)} decision problem:\\
\emph{Input:} A set $P$ of points in the plane and a constant $w>0$.\\
\emph{Output:} Whether there exists an Euclidean $t$-spanner for $P$ of weight at most $w$.
\end{myproblem}

\begin{myproblem}\label{prob:MWS$t$}
The \emph{\textsc{Minimum Weight $t$-Spanner} (MWS$t$)} problem:\\
\emph{Input:} A set $P$ of points in the plane.\\
\emph{Output:} A minimum weight Euclidean $t$-spanner for $P$.
\end{myproblem}


We show that for every real value $t>1$ the MWS$t$ and the LW$t$ problems are NP-hard.
This is done by a reduction 
from the \textsc{Partition} problem, defined as follows.
\begin{myproblem}\label{prob:Partition}
The \emph{\textsc{Partition}} problem:\\
\emph{Input:} A set $X=\{x_1,...,x_n\}$ of $n$ positive integers with even $\sum_{x \in X}x=R$.\\
\emph{Output:} Whether there exists a subset $X'\subset X$ such that $\sum_{x \in X'} x =R/2$.
\end{myproblem}
Klein and Kutz~\cite{KleinK06} have proved that the \textsc{Dilation Graph} 
and the \textsc{Plane Dilation Graph} problems are NP-hard by a reduction 
from the \textsc{Partition} problem as well.
The \textsc{Dilation Graph} problem (resp. the \textsc{Plane Dilation Graph} problem)
ask whether there exists a $t$-spanner (resp.  a plane $t$-spanner) with at most $m$ edges for a given 
set of points $P$, an integer $m$, and a real value $t>1$.
Note that here $t$ is part of the input. Their reduction returns an instance with $t=7$
for every instance of the \textsc{Partition} problem.


Various minimization problems of different parameters such as weight and number of edges 
of $t$-spanners have been proved to be NP-hard.
Cai et al. have proved in~\cite{Cai94} that for $t \geq 4$ the problem of determining the existence of 
a $t$-spanner with at most $m$ edges for an unweighted graph is NP-hard.
This implies that the problem of finding a minimum
$t$-spanner of weight at most $w$ for a weighted graph is NP-hard as well.
However, no similar conclusions regarding the NP-hardness of the problem 
in our geometric setting, i.e., of MWS$t$ and LWS$t$ can be deduced. 

Brandes and Handke have introduced a related problem of finding a minimum weight planar $t$-spanner
for weighted graphs. They have established it is NP-hard for $t>1$ by modifying the proof in~\cite{CaiC95}
of the NP-hardness of the tree $t$-spanner problem for weighted graphs.
A variation of the minimum weight planar $t$-spanner problem adjusted to the geometric 
setting discussed in this paper should restrict the underlying graph to be the complete Euclidean graph.
The appropriate decision and optimization problems are defined as follows.
\begin{myproblem}\label{prob:LWPS$t$}
The \emph{\textsc{Low Weight Plane $t$-Spanner} (LWPS$t$)} decision problem:\\
\emph{Input:} A set $P$ of points in the plane and a constant $w>0$.\\
\emph{Output:} Whether there exists an Euclidean plane $t$-spanner for $P$ of weight at most $w$.
\end{myproblem}

\begin{myproblem}\label{prob:MWPS$t$}
The \emph{\textsc{Minimum Weight Plane $t$-Spanner} (MWPS$t$)} decision problem:\\
\emph{Input:} A set $P$ of points in the plane and a constant $w>0$.\\
\emph{Output:} A minimum weight Euclidean Plane $t$-spanner for $P$.
\end{myproblem}
The reductions presented in this paper prove that the LWPS$t$ 
and MWPS$t$ problems are NP-hard for every $t>1$.

Regarding geometric graphs, Gudmundsson and Smid consider in~\cite{GudmundssonS09}
the problem of deciding whether a given geometric graph (not necessarily the complete graph)
contains a $t$-spanner with at most $m$ edges and prove it is NP-hard for every $t>1$.
This implies that the problem of finding a spanning graph of a given geometric graph
with at most $m$ edges and minimum dilation is also NP-hard.
In this paper we consider the minimality of spanners in terms of weight
and not in terms of the number of edges as addressed in this problem.
In addition, we restrict the underlying graph to be the complete Euclidean graph.
Therefore, we suggest the following variation.
\begin{myproblem}\label{prob:MDG}
The \emph{\textsc{Minimum Dilation graph} (MDG)} decision problem:\\
\emph{Input:} A set $P$ of points in the plane and a constant $w>0$.\\
\emph{Output:} A minimum dilation Euclidean graph for $P$ of weight at most $w$.
\end{myproblem}
From the NP-hardness of the LWS$t$ problem proved in this paper, 
we may deduce the NP-hardness of the MDG problem.

Note that although all the problems presented here refer to the restricted case 
where the underlying graph is the complete Euclidean graph, 
our results, obviously, apply to the modified problems where the underlying graph is
a general geometric graph (not necessarily the complete graph) and 
where the underlying graph is a general weighted graph (not necessarily geometric).

The rest of the paper is organized as follows.
In section~\ref{sec:defs} we define new terms and make some technical observations 
to be used in the reductions proofs.
The reductions are presented and proved in section~\ref{sec:red}. 
First the reductions idea is outlined and later described in more detail,
where the reduction for $t \geq 2$ is given in subsection~\ref{subsec:t>2}
and the reduction for $1<t<2$ is given in subsection~\ref{subsec:1<t<2}.


\section{Definitions and Technical Lemmas}\label{sec:defs}

\begin{definition}
Given a path $\{p,s,q\}$, a $t$-\emph{shortcut} refers to the addition of the edge $\{p,q\}$
and possibly a removal of one of the edges $\{p,s\}$ or $\{s,q\}$,
as long as the obtained graph is a $t$-spanner for $\{p,s,q\}$.
\end{definition}

\begin{definition}\label{def:effi}
Given a path $Q=\{p,s,q\}$ and
a \emph{$t$-shortcut} that results in a graph $G$, 
\begin{itemize}
\item the \emph{benefit} of the \emph{t-shortcut} is defined as $\delta_Q(p,q)-\delta_{G}(p,q)=|ps|+|sq|-|pq|$,
\item the \emph{cost} of the \emph{t-shortcut} is defined as $weight(G)-weight(Q)$, and
\item the \emph{efficiency} of the \emph{t-shortcut} is defined as the ratio between its \emph{benefit} and its \emph{cost},
i.e., $\frac{\delta_Q(p,q)-\delta_{G}(p,q)}{weight(G)-weight(Q)}$.
\end{itemize}
\end{definition}


\begin{lemma}\label{lem:isos}
Given two paths $Q=\{p,s,q\}$ and $Q'=\{p',s',q'\}$, such that $|ps|=|sq|$, $|p's'|=|s'q'|$,
and $\angle(psq)<\angle(p's'q')$ (see Fig.~\ref{fig:isos_obt}), 
and let $e$ and $e'$ denote the two efficiencies of the most efficient $t$-shortcuts in $Q$ and $Q'$, respectively, then $e>e'$.
\end{lemma}

\begin{proof}
By the fact that $\angle(psq)<\angle(p's'q')$, we conclude 
\begin{eqnarray}\label{eq:isos1}
\frac{|ps|}{|pq|}>\frac{|p's'|}{|p'q'|}.
\end{eqnarray}
The most efficient $t$-shortcut in $Q$ obviously includes the addition of $\{p,q\}$, 
and if $|ps|+|pq|\leq t|qs|$ it also includes the removal of $\{q,s\}$ (or alternatively $\{p,s\}$). 

If indeed $|ps|+|pq|\leq t|qs|$, then 
\begin{eqnarray*}
e &=& \frac{|ps|+|sq|-|pq|}{|pq|-|sq|} 
  \;=\; 1/(\frac{|pq|}{|ps|}-1)-1  \\
	&>^{(\ref{eq:isos1})}& 1/(\frac{|p'q'|}{|p's'|}-1)-1 
	\;=\; \frac{|p's'|+|s'q'|-|p'q'|}{|p'q'|-|s'q'|} \;\geq\; e'.
\end{eqnarray*}
Thus, $e>e'$ as required. 

Otherwise, $|ps|+|pq|> t|qs|$ and therefore 

\begin{eqnarray*}
 1+\frac{|pq|}{|qs|}> t 
 \;\overset{(\ref{eq:isos1})}\Rightarrow \; 1+\frac{|p'q'|}{|q's'|}> t 
 \;	\Rightarrow \; |p's'|+|p'q'|> t|q's'|
\end{eqnarray*}
and neither $\{q,s\}$ nor $\{q',s'\}$ are removed from the respective most efficient $t$-shortcuts. 
Hence, 
\begin{eqnarray*}
e&=& \frac{|ps|+|sq|-|pq|}{|pq|}
 \;=\;  \frac{2|ps|}{|pq|}-1\\
 &>^{(\ref{eq:isos1})}& \frac{2|p's'|}{|p'q'|}-1 
 \;=\; \frac{|p's'|+|s'q'|-|p'q'|}{|p'q'|} = e',
\end{eqnarray*}
and $e>e'$ as required. 

\end{proof}
%
%


\begin{lemma}\label{lem:obt}
Given a path $Q=\{p,s,q\}$ such that $|ps|<|sq|$, $\angle(psq)>\frac{\pi}{2}$, and 
$-\cos(\angle(psq))\geq \frac{1}{k+1}$ for a positive value $k$ (see Fig.~\ref{fig:isos_obt}), 
the efficiency of a $t$-shortcut in $Q$ is less than $k$.
\end{lemma}

\begin{proof}
Let $e$ denote the efficiency of the most efficient $t$-shortcut in $Q$ and let $a=|sq|$, $b=|ps|$, $c=|pq|$, and $\alpha=\angle(psq)$.
Note that $e \leq \frac{a+b-c}{c-a}$.

Since it is given that $-\cos\alpha \geq \frac{1}{k+1}$,
by the cosines law we have
$$c^2 = a^2+b^2-2ab\cos\alpha \geq a^2+b^2+\frac{2ab}{k+1} > a^2+\frac{b^2}{(k+1)^2}+\frac{2ab}{k+1}= (a+\frac{b}{k+1})^2.$$
Thus, we receive $c > a+\frac{b}{k+1}$  
and hence $e \leq\frac{a+b-c}{c-a} < k$.

\end{proof}
\begin{figure}[htb]
    \centering
        \includegraphics[width=1\textwidth]{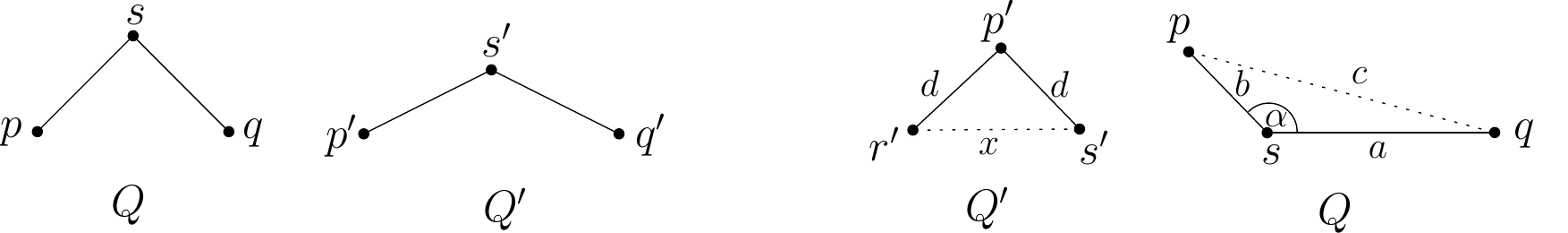}
    \caption{On the left, the paths $Q$ and $Q'$ as defined in Lemma~\ref{lem:isos};
    on the right, the paths $Q$ and $Q'$ as defined in Lemma~\ref{lem:obt} and Corollary~\ref{cor:obt} .}
    \label{fig:isos_obt}
\end{figure}


\begin{corollary}\label{cor:obt}
Given four points $\{r',s',s,q\}$ on a line and two additional points $p$ and $p'$,
such that $|r'p'|=|p's'|$, $|ps|<|sq|$, and the two edges $(p,s)$ and $(p',s')$ are parallel,
we consider the two paths $Q'=\{r',p',s'\}$ and $Q=\{p,s,q\}$ (see Fig.~\ref{fig:isos_obt}). 
Let $e'$ and $e$ denote the two efficiencies of the most efficient $t$-shortcuts in $Q'$ and $Q$, respectively, then $e'>e$.
\end{corollary}

\begin{proof}
Let $x=|r's'|$ and $d=|r'p'|=|p's'|$, then $e' \geq \frac{2d-x}{x}=2d/x-1$.
By the sines law, $\angle(rps)=2\arcsin(\frac{x}{2d})$.
Since the angle $\angle(psq)$, denoted by $\alpha$, equals to the 
exterior angle of the isosceles triangle $\triangle(rps)$, we have
$\alpha = 2\arcsin(\frac{x}{2d})+\pi/2-\arcsin(\frac{x}{2d})/2=\pi-\arccos(\frac{x}{2d})$ and $\cos(\alpha) = -\frac{x}{2d}$.
By Lemma~\ref{lem:obt}, we receive $e<2d/x-1 \leq e'$.
\end{proof}


\section{The reduction}\label{sec:red}

In this section we show that the LWS$t$ decision problem (Problem~\ref{prob:LWS$t$})
and the MWS$t$ optimization problem (Problem~\ref{prob:MWS$t$}) 
are NP-hard for every constant $t>1$.
We prove the NP-hardness of the LWS$t$ problem by a reduction 
from the \textsc{Partition} problem (Problem~\ref{prob:Partition}).
The NP-hardness of the MWS$t$ problem 
follows from an obvious reduction from the appropriate 
decision problem (the LWS$t$ problem).
We propose different reductions for $1<t<2$ and for $t \geq 2$;
however, both follow the same core ideas.

Given an instance $X=\{x_1,...,x_n\}$ with $\sum_{x \in X} x =R$ for the PARTITION problem,
both reductions output a weight $w$ and a set $P$ of linear size in $n$ 
that is composed of points distributed along a path 
connecting $n$ triples of isosceles triangles' vertices.
Each isosceles triangle gadget is associated with a value $x_i$ among $x_1,...,x_n$
and the length of its edges are derived from the value $x_i$. 
The distances between adjacent points on the path connecting all triangles
are derived from the sum $R$.

In order to prove the correctness of the reduction, namely, 
that a subset $X'\subset X$ with $\sum_{x \in X'} x =R/2$ exists 
iff there exists a $t$-spanner for $P$ of weight at most $w$, 
we use the same core method in both cases.
First, we observe that the minimum connected graph over $P$ forms a path
$\{p=p_1,...,p_m=q\}$ of weight (length) $w-c \cdot R= t|pq| + l \cdot R$. 
Meaning, a total shortening of $l \cdot R$ should be performed on the path
at the cost of at most $c\cdot R$.

We show that any shortening can be considered as a set of independent $t$-shortcuts
and that the most \emph{efficient} (as defined in Definition~\ref{def:effi}) 
$t$-shortcuts are those that involve adding a base of
an isosceles triangle gadget and their efficiency is exactly $l/c$.
We refer to such $t$-shortcuts as \emph{gadget $t$-shortcuts}.
Therefore, no other shortcut may be applied and 
the distance between $p$ and $q$ can be decreased to $t|pq|$ iff there  
exists a set of gadget $t$-shortcuts with a total benefit of $l\cdot R$.
Our construction ensures that a set of gadget $t$-shortcuts 
with a total benefit of $l\cdot R$ exists iff
there exists a subset $X'\subset X$ such that $\sum_{x \in X'} x =R/2$
and applying those $t$-shortcuts creates a $t$-spanner for $P$.
Moreover, the $t$-spanner obtained by applying those $t$-shortcuts is a plane graph.
Therefore, our reductions also prove the NP-hardness of the LWPS$t$ problem (Problem~\ref{prob:LWPS$t$})
and the MWPS$t$ problem (Problem~\ref{prob:MWPS$t$}).

The reduction for $t \geq 2$ is a bit simpler and therefore 
described first, in the next subsection,
followed by the reduction for $1<t< 2$ in subsection~\ref{subsec:t>2}.


\subsection{The reduction for $t \geq 2$}\label{subsec:t>2}

Given a valid input for the \textsc{Partition} problem $X=\{x_1,...,x_n\}$ with $\sum_{x \in X}x = R$,
the reduction outputs a valid input for the SW$t$ problem $(P,w)$ as elaborated next.
The points of $P$ are placed roughly along 3 sides of an axis-parallel rectangle,
the left side, the right side, and the top edge as illustrated in Fig.~\ref{fig:bigt}. 
The left and right sides are of length $R/2(t(n+2)-n-2\frac{1}{3})$; 
each is sampled by $\left\lceil{t(n+2)-n-2\frac{1}{3}}\right\rceil$ points
with regular spacing of $R/2$ (possibly except for the bottom point).
The top edge is actually a horizontal component of width $R(n+2)$, 
consists of $n$ isosceles triangle gadgets and connecting points.
The top edge itself is sampled by $2n$ endpoints of the $n$ isosceles triangles' bases, 
located among $n+1$ segments of length $R$,
halved by a middle point. The additional vertex of each triangle is located above the rectangle.
The $i$-th isosceles triangle gadget has sides of length $\frac{5}{6}x_i$,
a base of length $x_i$, and subtended angle of $2\arcsin(\frac{3}{5}) < \pi/2$.
Overall, we get $|P|=4n+3+2 \left\lceil{t(n+2)-n-2\frac{1}{3}}\right\rceil$.
One can verify that $P$ is of polynomial size in $n$ and can be represented in 
polynomial size in $n$ and $\log R$, and thus can be generated in polynomial time.
\begin{figure}[htb]
    \centering
        \includegraphics[width=0.7\textwidth]{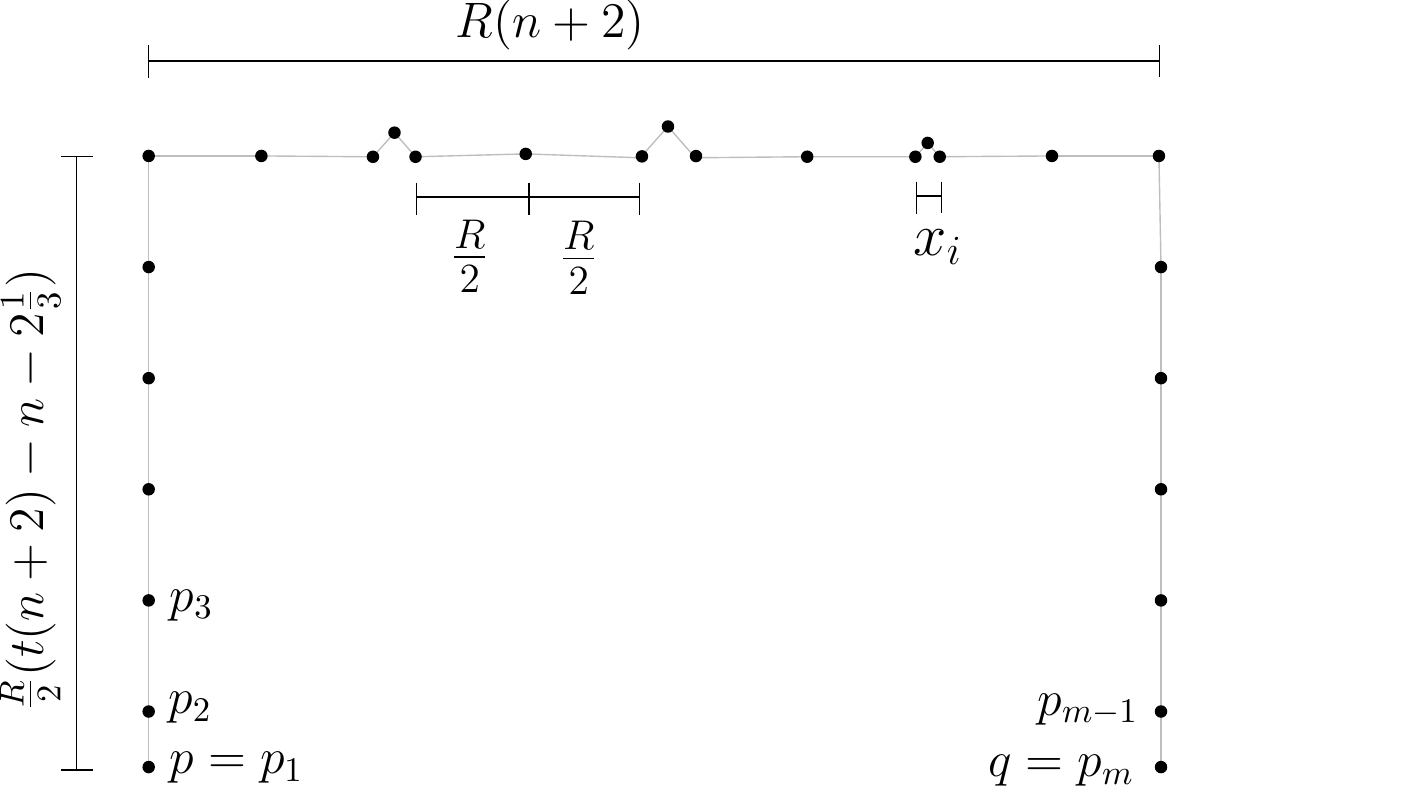}
    \caption{The set of points $P$ and its minimum connected graph as defined in the reduction for $t\geq2$.}
    \label{fig:bigt}
\end{figure}

The weight bound $w$ is defined a bit differently for \smallt{} and \bigt{}.
For \smallt{} we define $w=R(t(n+2)+\frac{5}{6})$
and for  \bigt{}, $w=R(t(n+2)+\frac{5}{12})$.
We are now ready to prove the correctness of the reduction; namely, 
that a subset $X'\subset X$ with $\sum_{x \in X'} x =R/2$ exists 
iff there exists a $t$-spanner for $P$ of weight at most $w$.

First, observe that the minimum connected graph over $P$ forms a path
$\{p=p_1,...,p_m=q\}$ as depicted in Fig.~\ref{fig:bigt}
of weight $2\cdot R/2(t(n+2)-n-2\frac{1}{3}) + R(n+1) + \frac{5}{3}R= R(t(n+2)+\frac{1}{3})$.
Since $|pq|=R(n+2)$, this path is $R/3$ units longer than a legal
$t$-spanning path from $p$ to $q$;
however, it is lighter then the given weight $w$.
For \smallt{} the remaining weight is $R/2$
and for \bigt{} the remaining weight is $R/12$.
Therefore, a total shortening of $R/3$ should be performed on the path
at the cost of at most $R/2$ or $R/12$ for \smallt{} and \bigt{} respectively.

Assuming there exists a $t$-spanner for $P$ of weight at most $w$,
let $E^+$ and $E^-$ be the sets of edges over $P$ that should be added and removed 
from $MST(P)$, respectively, in order to obtain a $t$-spanner $G=(P,E)$ 
that minimizes $weight(G) \leq w$.
We make some observations regarding the edge set $E^+$.

\begin{myclaim}\label{cl:short_edges}
The edge set $E^+$ contains only edges of length at most $R$.
\end{myclaim}

\begin{proof}
Since the initial graph is a minimum spanning tree and $G$ must be connected, $|E^-| \leq |E^+|$; 
moreover, there exists an injective function from $E^-$ to $E^+$ 
that maps every edge $e^- \in E^-$ to a longer edge $e^+ \in E^+$ 
whose addition creates a cycle in $MST(P)$ and enables the removal of $e^-$. 
Assume towards contradiction that $E^+$ contains an edge $e^+$ longer than $R$.
Since all edges in $MST(P)$ are of length at most $\max\{R/2,\max_{1\leq i\leq n}\{5x_i/6 \}\}\leq^{(*)} R/2$,
we have that $$\sum_{e \in E^+}{weight(e)} - \sum_{e \in E^-}{weight(e)} > R/2$$
and thus $weight(G)>w$. Inequality (*) holds under the legitimate assumption that all the elements in $X$ are
smaller than $R/2$.
\end{proof}

According to the above claim, only edges of length at most $R$ can be added to $MST(P)$.
There are exactly three types of such edges 
(while ignoring edges that overlaps subpaths in $MST(P)$, 
which are useless for shortening the distance between $p$ and $q$):

\begin{enumerate}
\item A base edge of an isosceles triangle gadget on the top edge.
\item An edge connecting a top vertex of an isosceles triangle on the top side 
			with a point on the top side itself that closes a cycle of size three.
\item An edge connecting the left (resp. right) edge with the top edge 
			and closes a cycle of size three.
\end{enumerate} 

\begin{observation}\label{obs:disjoint}
Any two edges in $E^+$ close edge-disjoint cycles in $MST(P)$.
\end{observation}

\begin{proof}
The only pair of edges that may violate this condition is a pair of an edge of type 1 
and an edge of type 2.
Assume towards contradiction that there are two such edges, 
$e_1$ of type 1 and $e_2$ of type 2, in $E^+$.
Due to the minimality of $weight(G)$ there is necessarily a $t$-spanning path 
connecting two points in $G$ that contains $e_2$. 
Note that between every two points on the top side of the rectangle there is 
a $t$-spanning path in $MST(P)$; hence, $e_2$ is necessarily contained in 
a $t$-spanning path connecting two points on different sides of the rectangle
that passes through the top side.
According to the same arguments, the same holds for $e_1$. 
Meaning, $e_1$ and $e_2$ are both contained in the shortest path connecting
the endpoints of the top side of the rectangle (see Fig.~\ref{fig:zigzag}).
However, replacing $e_1$ and $e_2$ with the induced subgraph of $MST(P)$ 
over the endpoints of $e_1$ and $e_2$ results in a shorter path.
Therefore, the addition of $e_1$ and $e_2$ to $E^+$ is not necessary for
creating any $t$-spanning path and thus contradicts the minimality of $weight(G)$.
\end{proof}

\begin{figure}[htb]
    \centering
        \includegraphics[width=0.4\textwidth]{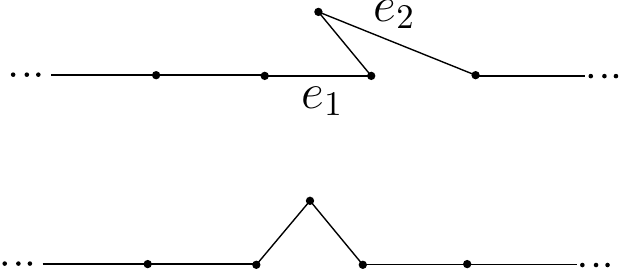}
    \caption{On the top, part of the $t$-spanning path between two points in $G$ 
    				containing $e_1$ and $e_2$, as discussed in the proof of Observation~\ref{obs:disjoint}.
    				On the bottom, part of the $t$-spanning path between the same two points in $MST(P)$.}
    \label{fig:zigzag}
\end{figure}

By Observation~\ref{obs:disjoint} we conclude that $G$ can be obtained from $MST(P)$ by applying 
a set of $t$-shortcuts in edge-disjoint subpaths of $MST(P)$.
As we have already observed, each $t$-shortcut involves an addition of an edge of type 1, 2, or 3
and possibly a removal of an edge in the closed cycle. 
The addition of an edge of type $1$, together with a removal of an isosceles triangle side
for \bigt{}, we refer to as a \emph{gadget} $t$-shortcut.

\begin{myclaim}\label{cl:tri}
The most efficient possible $t$-shortcuts are \emph{gadget} $t$-shortcuts. 
Their efficiency is $\frac{2}{3}$ for  \smallt{} and $4$ for \bigt{}.
\end{myclaim}

\begin{proof}
First, note that the removal of a triangle side in a gadget $t$-shortcut for \bigt{}
is indeed possible in the aspect of dilation.
Now let us examine the efficiency of a gadget $t$-shortcut.
The benefit of the $i$-th gadget $t$-shortcut is $\frac{5}{3}x_i-x_i = 2x_i/3$.
For \smallt{} its cost is $x_i$ and for \bigt{} its cost is $1\frac{5}{6}x_i-\frac{5}{3}x_i=x_i/6$; 
hence, its efficiency is indeed $2/3$ and $4$ for \smallt{} and \bigt{}, respectively.

Next, we observe every conceivable $t$-shortcut and show it is less efficient than a gadget $t$-shortcut.

\begin{itemize}

\item Addition of an edge of type 2 forms an obtuse triangle.
By Corollary~\ref{cor:obt}, this $t$-shortcut is less efficient than the gadget $t$-shortcut.

\item Addition of an edge of type 3 forms an isosceles right triangle 
and by Lemma~\ref{lem:isos} the gadget $t$-shortcut is more efficient.

\end{itemize}
\end{proof}

According to Claim~\ref{cl:tri}, only the gadget $t$-shortcuts have efficiency equal to
the ratio between the required shortening of the path between $p$ and $q$ and the remaining weight.
Therefore, those are the only shortcuts that can be applied.
Moreover, in order to achieve a $R/3$ shortening of the path connecting $p$ and $q$
without exceeding the weight bound, a set of gadget $t$-shortcuts 
with a total benefit of exactly $R/3$ should be applied. 

\begin{lemma}\label{lem:red}
A subset $X'\subset X$ with $\sum_{x \in X'} x =R/2$ exists 
iff there exists a $t$-spanner for $P$ of weight at most $w$.
\end{lemma} 

\begin{proof}

[$\Leftarrow$] Assume towards contradiction that 
there exists a $t$-spanner for $P$ of weight at most $w$, however, 
a subset $X'\subset X$ with $\sum_{x \in X'} x =R/2$ does not exist. 
According to our construction, this implies that a set of gadget $t$-shortcuts 
with a total benefit of exactly $R/3$ does not exist. 
By the aforementioned analysis, this means that there is no graph over $P$
of weight at most $w$ in which there exists a path connecting $p$ and $q$ of length at most $t|pq|$.
This contradicts the existence of a $t$-spanner for $P$ of weight at most $w$.

[$\Rightarrow$] Assume that a subset $X'\subset X$ with $\sum_{x \in X'} x =R/2$ exists.
Let $G=(P,E)$ denote the graph obtained by applying the gadget $t$-shortcuts
in the triangles that correspond to the elements in $X'$.
We show that $G$ admits a $t$-spanner for $P$, i.e., 
for every two points $u$ and $v$ in $P$, $\delta_G(u,v)\leq t|uv|$.
We have already shown that there exists a $t$-spanning path between $p$ and $q$ in $G$.
Next we consider all other pairs of points $\{u,v\}\neq \{p,q\}$ in $P$:
\begin{enumerate}

\item One of $\{u,v\}$ is on the right side of the rectangle and the other is on the left side:
$\delta_G(u,v) \leq \delta_G(p,q)$ and $|uv| \geq |pq|$, hence $\delta_G(u,v)/|uv| \leq \delta_G(p,q)/|pq| = t$.

\item One of $\{u,v\}$ is on the right or left side of the rectangle and the other is on the top side:
Assume w.l.o.g. that $u$ is on the left side and $v$ is on the top side (see Fig.~\ref{fig:connections}). 
Let $a$ denote the left-top corner point of the rectangle.
We have
\begin{eqnarray*}
\delta_G(u,v)/|uv|&=& (\delta_G(u,a)+\delta_G(a,v))/|uv|\\
							&\leq^{(*)}& \max\{(|ua|+|av|(\frac{4}{5}+\frac{1}{5}\cdot 1\frac{5}{6}/\frac{1}{2}))/|uv|,\\ 
													&& \;\;\;\;\;\;\;\;\;(|ua|+|av|(\frac{2}{3}+\frac{1}{3}\cdot \frac{5}{3}))/|uv|\}\\
							&\leq& (|ua| + 1\frac{8}{15}|av|)/\sqrt{|ua|^2+|av|^2} \\
							&<^{(**)}& 2 \leq t|uv|
\end{eqnarray*}
Inequality (*) holds under the assumption that all the elements in $X$ are
smaller than $R/2$ and due to the notion that
every triangle gadget base is followed by two edges of length $R/2$ on the left.
Inequality (**) holds since the values of the function $(x + 1\frac{8}{15}y)/\sqrt{x^2+y^2}$
are smaller than $2$ for every $x$ and $y$.

\item Both $u$ and $v$ are on the top side of the rectangle:
As noted before, if $u$ and $v$ are the endpoints of a side of 
an isosceles triangle gadget we have $\delta_G(u,v) \leq t|uv|$.
Let $u'$ and $v'$ be the projections of $u$ and $v$ on the line 
containing the top side of the rectangle (see Fig.~\ref{fig:connections}).
By the same arguments as in the previous case, 
we have 
$$\delta_G(u,v) \leq (\frac{2}{3}+\frac{1}{3} \cdot 1\frac{5}{6}/\frac{1}{2})|u'v'|< 2|uv| \leq t|uv|.$$

\item Both $u$ and $v$ are on the left or the right side of the rectangle:
We have $\delta_G(u,v)=|uv| < t|uv|$.
\end{enumerate}

\end{proof}
\begin{figure}[htb]
    \centering
        \includegraphics[width=0.9\textwidth]{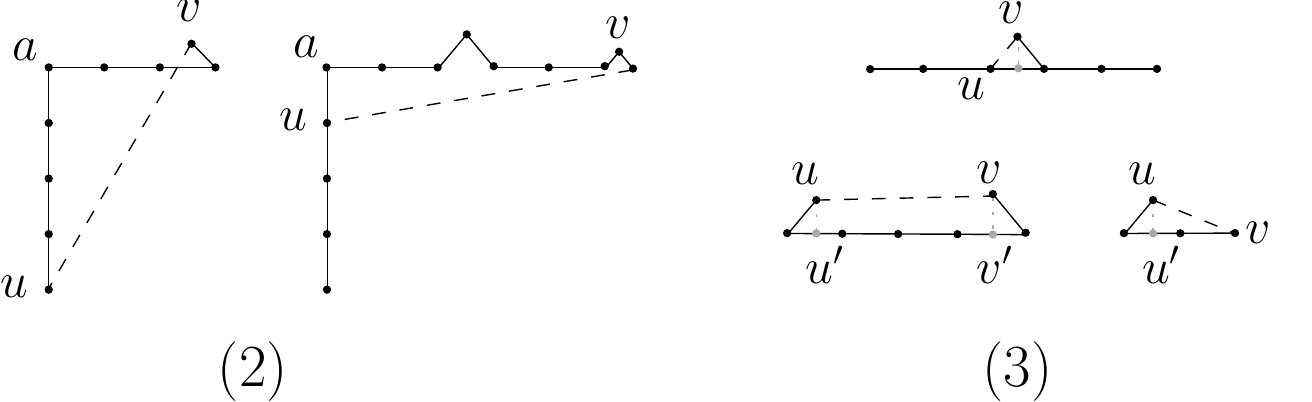}
    \caption{On the left, illustrations for case 2 in the proof of Lemma~\ref{lem:red};
    on the right, illustrations for case 3 in the same proof. }
    \label{fig:connections}
\end{figure}
Overall, we proved a proper reduction for $t \geq 2$.

%
\subsection{The reduction for $1<t<2$}\label{subsec:1<t<2}

We reduce from the PARTITION problem
and define the reduction output $(P,w)$ as follows.
The weight bound is $w=R(n+t+\frac{3}{2}+(n+\frac{3}{2})\frac{3t}{2})$ and the set of points $P$
resembles the one defined for $t\geq 2$. 
We still have a horizontal component, and left and right sides;
however, those are no longer perpendicular to the horizontal component.
The top right angles are enlarged by $\alpha_t=\arcsin(\frac{2}{3t^2}+\frac{1}{3t})$ 
and the rectangle shape turns into an isosceles trapezoid shape (see Fig.~\ref{fig:smallt}).

As before, the horizontal component is composed of a horizontal segment of length $R(n+2)$ 
sampled by $2n$ endpoints of $n$ isosceles triangles' bases, 
located among $n+1$ segments of length $R$
halved by a middle point. 
However, here the $i$-th isosceles triangle has sides of length $\frac{t}{2}x_i$, 
a base of length $x_i$, and thus an angle of $2\arcsin(1/t)$.

The right and left sides are of length $R/2(n+\frac{3}{2})\frac{3t}{2}$;
each is sampled by $\left\lceil (n+\frac{3}{2})\frac{3t}{2}\right\rceil$
points with regular distances of $R/2$ (possible except for the bottom point).

Overall, we get $|P|=4n+3+2\left\lceil (n+\frac{3}{2})\frac{3t}{2}\right\rceil$.
We denote by $p$ and $q$ the leftmost and rightmost points in $P$, respectively,
and by $p'$ and $q'$ the leftmost and rightmost points of the horizontal component, respectively.

The reduction as presented above might not be computable by a Turing machine in polynomial time,
since there might be points in $P$ with coordinates whose representation requires large number of bits.
This issue is resolved by a minor change as described in the end of this subsection.

We prove the correctness of the reduction, namely, 
that a subset $X'\subset X$ with $\sum_{x \in X'} x =R/2$ exists 
iff there exists a $t$-spanner for $P$ of weight at most $w$,
by following the same guidelines as for the proof for $t \geq 2$.

Note that the minimum connected graph over $P$ forms a path
$Q=\{p=p_1,...,p_m=q\}$ as depicted in Fig.~\ref{fig:smallt}
of length 
$$2R/2(n+\frac{3}{2})\frac{3t}{2} + (n+1)R + tR =R(n+t+1+(n+\frac{3}{2})\frac{3t}{2}).$$
Thus, the remaining weight is $R/2$.
We would like to examine how much should this path be shortened 
in order for it to be a proper $t$-spanning path between $p$ and $q$.
A proper $t$-spanning path between $p$ and $q$ is of length at most 
\begin{eqnarray*}
t|pq| &=& t(2\frac{R}{2}(n+\frac{3}{2})3t/2\sin\alpha_t + (n+2)R)\\
			&=& Rt(R(n+\frac{3}{2})\frac{3t}{2}\sin(\arcsin(\frac{2}{3t^2}+\frac{1}{3t}))+ n + 2)\\
			&=& Rt((n+\frac{3}{2})\frac{3t}{2}(\frac{2}{3t^2}+\frac{1}{3t}+ n + 2))\\
			&=& R((n+\frac{3}{2})\frac{3t}{2}(1-(\frac{2}{3}-\frac{2}{3t}))+ tn + 2t)\\
			&=& R((n+\frac{3}{2})\frac{3t}{2} + n + \frac{3}{2} + \frac{t}{2}).
\end{eqnarray*}
Therefore, a total shortening of $(t-1)R/2$ should be performed on the path
at the cost of at most $R/2$.
\begin{figure}[htb]
    \centering
        \includegraphics[width=0.8\textwidth]{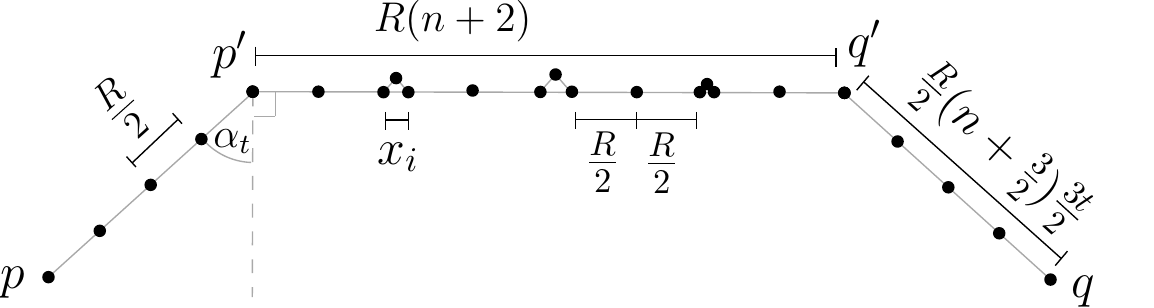}
    \caption{The set of points $P$ is depicted in black 
    and its minimum connected graph is depicted in gray as defined in the reduction for $1<t<2$.}
    \label{fig:smallt}
\end{figure}

Assuming there exists a $t$-spanner for $P$ of weight at most $w$,
let $E^+$ and $E^-$ be the sets of edges over $P$ that should be added and removed 
from $MST(P)$, respectively, in order to obtain a $t$-spanner $G=(P,E)$ 
that minimizes $weight(G) \leq w$.
By Claim~\ref{cl:short_edges} (which, by the same arguments, holds for this reduction as well), 
$E^+$ contains only edges of length at most $R$.
There are exactly four types of such edges 
(while ignoring edges that overlap subpaths in $MST(P)$, 
which are obviously redundant):

\begin{enumerate}
\item A base edge of an isosceles triangle gadget on the horizontal component.
\item An edge connecting a top vertex of an isosceles triangle on the horizontal component 
			with a point on the horizontal segment and closes a cycle of size three.
\item An edge connecting the left (resp. right) side with the horizontal component 
			and closes a cycle of size three.
\end{enumerate} 

One can verify that Observation~\ref{obs:disjoint} holds here as well and
we have that $G$ can be obtained from $MST(P)$ by applying 
a set of $t$-shortcuts that involve an addition of an edge of type 1, 2, or 3
and possibly removal of an edge in the closed cycle.
Here, the term \emph{gadget $t$-shortcut} refers to
the addition of an edge of type $1$. 
No removal of an edge is possible for $1<t<2$.


\begin{myclaim}\label{cl:tri2}
The most efficient possible $t$-shortcuts are \emph{gadget} $t$-shortcuts. 
Their efficiency is $t-1$.
\end{myclaim}

\begin{proof}
First we examine the efficiency of a gadget $t$-shortcut.
The benefit of the $i$-th gadget $t$-shortcut is $tx_i-x_i=(t-1)x_i$,
its cost is $x_i$, and hence its efficiency is indeed $t-1$.

We now examine every conceivable $t$-shortcut and show that it is 
less efficient than a gadget $t$-shortcut.

\begin{itemize}

\item Addition of an edge of type 2 forms an obtuse triangle. 
By Corollary~\ref{cor:obt}, this $t$-shortcut is less efficient than the gadget $t$-shortcut.

\item Addition of an edge of type 3 forms an isosceles triangle with angle \\
$\pi/2+\arcsin(\frac{2}{3t^2}+\frac{1}{3t})$. This angle is greater than $2\arcsin(1/t)$
for every $t>1$ and by Lemma~\ref{lem:isos}, we conclude that the gadget $t$-shortcut is more efficient.

\end{itemize}
\end{proof}

By Claim~\ref{cl:tri2}, in order to achieve a $(t-1)R/2$ shortening 
of the path connecting $p$ and $q$,
a set of gadget $t$-shortcuts 
with a total benefit of exactly $(t-1)R/2$ should be applied.
This notion, together with the following Lemma, 
serve us in the proof of Lemma~\ref{lem:red_smallt}.


\begin{lemma}\label{lem:triStretch}
Let $\triangle(psq)$ be an isosceles triangle with $|ps|=|sq|$ 
and let $q'$ be a point on $\overline{sq}$ such that $|pq| \geq |pq'|$,
then $\frac{|ps|+|sq|}{|pq|} \geq \frac{|ps|+|sq'|}{|pq'|}$.
\end{lemma}

\begin{proof}
Let $\gamma = \angle(psq)$ and $\beta = \angle(pq'q)$,
then by the sines law, $|pq|=|ps|2\sin(\frac{\gamma}{2})$, \\
$|sq'|=|ps|\sin(\beta-\gamma)/\sin(\beta)$, and
$$|pq'|= |pq|\sin(\frac{\pi}{2}-\frac{\gamma}{2})/\sin(\beta)=
|ps|2\sin(\frac{\gamma}{2})\cos(\frac{\gamma}{2})/\sin(\beta).$$
Thus, 
\begin{eqnarray*}
 && \frac{|ps|+|sq|}{|pq|} \geq \frac{|ps|+|sq'|}{|pq'|}\\
 &\Leftrightarrow& \frac{2|ps|}{|ps|2\sin(\frac{\gamma}{2})} \geq 
 										\frac{|ps|+|ps|\sin(\beta-\gamma)/\sin(\beta)}{|ps|2\sin(\frac{\gamma}{2})\cos(\frac{\gamma}{2})/\sin(\beta)}\\
 &\Leftrightarrow& \frac{\cos(\frac{\gamma}{2})}{\sin(\beta)} \geq 
 										1+\frac{\sin(\beta-\gamma)}{\sin(\beta)}\\
 &\Leftrightarrow& 	\cos(\frac{\gamma}{2}) -	\sin(\beta) -	\sin(\beta-\gamma)	\geq 0.
\end{eqnarray*}

The last inequality indeed holds for every $0<\gamma,\beta<\pi$.
\end{proof}

\begin{figure}[htb]
    \centering
        \includegraphics[width=0.36\textwidth]{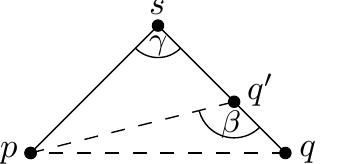}
    \caption{Illustration of Lemma~\ref{lem:triStretch}.}
    \label{fig:triangle_sf}
\end{figure}


\begin{lemma}\label{lem:red_smallt}
A subset $X'\subset X$ with $\sum_{x \in X'} x =R/2$ exists 
iff there exists a $t$-spanner for $P$ of weight at most $w$.
\end{lemma} 

\begin{proof}

[$\Leftarrow$] Assume towards contradiction that 
there exists a $t$-spanner for $P$ of weight at most $w$, however, 
a subset $X'\subset X$ with $\sum_{x \in X'} x =R/2$ does not exist. 
According to our construction, this implies that a set of gadget $t$-shortcuts 
with a total benefit of exactly $(t-1)R/2$ does not exist. 
This means there is no graph over $P$
of weight at most $w$ that contains a path connecting $p$ and $q$ of length at most $t|pq|$ 
in contradiction to the existence of a $t$-spanner for $P$ of weight at most $w$.

[$\Rightarrow$] Assume that a subset $X'\subset X$ with $\sum_{x \in X'} x =R/2$ exists.
Let $G=(P,E)$ denote the graph obtained by applying the gadget $t$-shortcuts
in the triangles that correspond to the elements in $X'$.
We show that $G$ admits a $t$-spanner for $P$, i.e., 
for every two points $u$ and $v$ in $P$, $\delta_G(u,v)\leq t|uv|$.
As we have already observed, there exists a $t$-spanning path between $p$ and $q$ in $G$.
Next we consider all other pairs of points $\{u,v\}\neq \{p,q\}$ in $P$:
\begin{enumerate}

\item One of $\{u,v\}$ is on the right side of the trapezoid 
and the other is on the left side (see Fig.~\ref{fig:smallt_sf+}):
Assume w.l.o.g. that $u$ is on the left side and $v$ is on the right side.
Let $s$ be the intersection point of the extensions of the segments $\overline{pp'}$ and $\overline{qq'}$
(see Fig.~\ref{fig:smallt_sf+}).
Let $\epsilon=|p's|+|sq'|-\delta_G(p',q')$,
we show that $(|us|+|sv|)/|uv| \leq (|ps|+|sq|)/|pq|$
and conclude,
$$\frac{\delta_G(u,v)}{|uv|} = \frac{|us|+|sv|-\epsilon}{|uv|} \leq 
\frac{|ps|+|sq|}{|pq|}-\frac{\epsilon}{|uv|} \leq \frac{|ps|+|sq|}{|pq|}-\frac{\epsilon}{|pq|}= t.$$

If $\overline{uv}$ is parallel to $\overline{pq}$, then 
due to the similarity of the triangles $\triangle(usv)$ and $\triangle(psq)$
we have $(|us|+|sv|)|uv| = (|ps|+|sq|)/|pq|$.

Otherwise, assume w.l.o.g. that $|up|<|vq|$ and let $v'$ be a point
on $\overline{qq'}$ such that $\overline{uv'}$ is parallel to $\overline{pq}$.
Note that $\angle(uvv')>\pi/2$ and thus, by Lemma~\ref{lem:triStretch} 
we have, $(|us|+|sv|)/|uv| \leq (|us|+|sv'|)/|uv'| \leq (|ps|+|sq|)/|pq|$.

\item One of $\{u,v\}$ is on the right or left side of the trapezoid 
and the other is on the horizontal component:
Assume w.l.o.g. that $u$ is on the left side and $v$ is on 
the horizontal component.
Recall our note regarding every triangle gadget base being
followed by two edges of length $R/2$ on the left and thus 
$\delta_G(p',v) \leq (2/3+t/3)|p'v|$.
By the cosines law we have 
$$|uv| \geq \sqrt{|up'|^2+|p'v|^2- 2|up'||p'v|\cos(\pi+\alpha_t)}.$$
(we do not use equality since $v$ may not lie in the horizontal segment itself, 
but on a top vertex of a triangle gadget,
and than $\angle(up'v)>\alpha_t$).
Hence, we receive
\begin{eqnarray*}
\delta_G(u,v) &=& \delta_G(u,p') + \delta_G(p',v)\\
							&\leq & |up'| + (\frac{2}{3}+\frac{t}{3})|p'v|\\
							&= & \sqrt{|up'|^2 + (\frac{2}{3}+\frac{t}{3})^2|p'v|^2 + 2|up'||p'v|(\frac{2}{3}+\frac{t}{3})}\\
							&\leq& \sqrt{t^2|up'|^2 + t^2|p'v|^2 - 2|up'||p'v|t^2\cos(\pi/2+\alpha_t)} \leq t|uv|						
\end{eqnarray*}

\item Both $u$ and $v$ are on the horizontal component:
Let $u'$ and $v'$ be the projections of $u$ and $v$ on the horizontal segment.
We have $\delta_G(u,v) \leq t|u'v'| \leq t|uv|$.

\item Both $u$ and $v$ are on the left or the right side of the trapezoid:
We have $\delta_G(u,v)=|uv| < t|uv|$.

\end{enumerate}

\end{proof}

\begin{figure}[htb]
    \centering
        \includegraphics[width=0.8\textwidth]{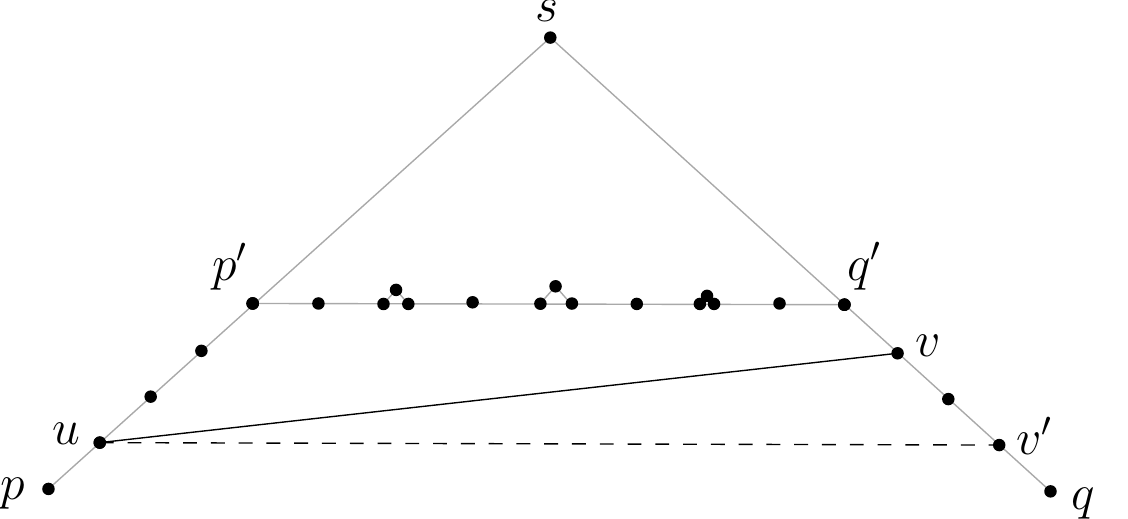}
    \caption{Illustration of cases 1 in the proof of Lemma~\ref{lem:red_smallt}.}
    \label{fig:smallt_sf+}
\end{figure}

The only thing left to prove is that the above reduction can be computed 
by a Turing machine in polynomial time.
Clearly, $P$ is of polynomial size in $n$; however, not all the points in $P$ 
can necessarily be represented in 
polynomial number of bits.
More precisely, setting the coordinates of the points on the horizontal
segment to integer values may
cause the values of other points' coordinates to have large representation size.
This issue can be easily solved by rounding the values of the coordinates
of the rest of the points to have up to $n$ digits after the decimal point as follows.
The y-coordinates of the triangle gadgets' top vertices are rounded down
(note that the x-coordinate has polynomial size representation), 
and for the points on the sides, the rounding is done in such a way that
increases the top angles of the trapezoid and does not increase their length.

Note that as $n$ increases, the efficiency of the gadget $t$-shortcuts
increases and tends to its initial value,
while the efficiency of the other $t$-shortcuts does not exceed its initial value.
Hence, for sufficiently large values of $n$, the absence of 
a proper partition of the reduction input $X$ guarantees that a $t$-spanner for $P$
of weight at most $w$ does not exists.
In case a proper partition of $X$ exists,
the suggested spanner over the output set $P$ is still of weight smaller than $w$
and has stretch factor at most $t$. 
Thus, for sufficiently large values of $n$, 
the reduction remains correct after the rounding.
Obviously, the PARTITION problem remains NP-hard after restricting $n$ to be larger
than some constant $n_0$. 

In conclusion, after presenting proper reductions for all values of $t$, the main theorem follows.

\begin{theorem}\label{theo:main}
The decision problems LWS$t$ and LWPS$t$, and the optimization problems MWS$t$, MWPS$t$, and MDG are NP-hard.
\end{theorem}

Although all the problems addressed in Theorem~\ref{theo:main} refer to the restricted case 
where the underlying graph is the complete Euclidean graph, 
our results, obviously, apply to the cases where the underlying graph is
a general geometric graph (not necessarily the complete graph) and 
where the underlying graph is a general weighted graph (not necessarily geometric).
This is stated in the following corollary.

\begin{corollary}
The variations of the decision problems LWS$t$ and LWPS$t$ and the optimization problems MWS$t$, MWPS$t$, and MDG 
addressing a geometric or general weighted underlying graph are NP-hard.
\end{corollary}

\bibliographystyle{abbrv}
\bibliography{ref}

\end{document}